\documentclass[journal,draftcls,onecolumn,12pt,twoside]{IEEEtranTCOM}
\normalsize
\ifCLASSINFOpdf
\else
\fi

\usepackage{amsmath}
\usepackage{amsthm}
\usepackage{graphicx}
\usepackage{tipa}
\usepackage{changes}
\usepackage{mathtools}
\usepackage{amssymb}
\usepackage[export]{adjustbox}
\usepackage{stfloats}
\newcommand*{\rom}[1]{\expandafter\@slowromancap\romannumeral #1@}
\usepackage{caption}
\ifCLASSOPTIONcompsoc
\usepackage[caption=false, font=normalsize, labelfont=sf, textfont=sf]{subfig}
\else
\usepackage[caption=false, font=footnotesize]{subfig}
\fi
\usepackage{tikz}
\usetikzlibrary{calc}
\newtheorem{remark}{Remark}
\newtheorem{claim}{Claim}

\newtheorem{theorem}{Theorem}
\newtheorem{lemma}{Lemma}

\DeclareFontFamily{U}{mathx}{\hyphenchar\font45}
\DeclareFontShape{U}{mathx}{m}{n}{
	<5> <6> <7> <8> <9> <10>
	<10.95> <12> <14.4> <17.28> <20.74> <24.88>
	mathx10
}{}

\begin{document}
	\title{A Secure Key Sharing Algorithm Exploiting Phase Reciprocity in Wireless Channels}
	
	\author{Shayan Mohajer Hamidi,~\IEEEmembership{Student Member,~IEEE}, Amir Keyvan Khandani,~\IEEEmembership{Member,~IEEE}, and Ehsan Bateni}

\markboth{}%
{Submitted paper}
	\maketitle
	\begin{abstract}
This article presents a secure key exchange algorithm that exploits reciprocity in wireless channels to share a secret key between two nodes $A$ and $B$. Reciprocity implies that the channel phases in the links $A\rightarrow B$ and $B\rightarrow A$ are the same. A number of such reciprocal phase values are measured at nodes $A$ and $B$, called shared phase values hereafter. Each shared phase value is used to mask points of a Phase Shift Keying (PSK) constellation. Masking is achieved by rotating each PSK constellation with a shared phase value. Rotation of constellation is equivalent to adding phases modulo-$2\pi$, and as the channel phase is uniformly distributed in $[0,2\pi)$, the result of summation conveys zero information about summands. To enlarge the key size over a static or slow fading channel, the Radio Frequency (RF) propagation path is perturbed to create several independent realizations of multi-path fading, each used to share a new phase value. To eavesdrop a phase value shared in this manner, the Eavesdropper (Eve) will always face an under-determined system of linear equations which will not reveal any useful information about its actual solution value. This property is used to establish a secure key between two legitimate users.
	\end{abstract}
	\begin{IEEEkeywords}
		Physical layer security, Channel reciprocity, secret key generation, 
	\end{IEEEkeywords}
	\IEEEpeerreviewmaketitle
	\section{introduction} \label{intro}
	\IEEEPARstart{T}{he} inherent broadcast nature of wireless communication makes it vulnerable to different types of attacks such as Eavesdropping. Traditionally, the data is secured by using classic encryption schemes, which entail functions and algorithms that provide a desired level of secrecy for communicating data. Most commonly used encryption methods rely on the computational hardness of some mathematical problems, {\it e.g.}, discrete logarithm. Nevertheless, due to continual advances in computer technology, and the discovery of new computational techniques such as Quantum computing, eavesdroppers are becoming more equipped and intelligent, and thus long-term effectiveness of such traditional techniques is questionable. As a result, there is a need for an extra level of secrecy to strengthen the security of communicating data in wireless communication.  

To overcome these challenges, in addition to traditional cryptographic techniques---which are applied at the third or higher layers of communication protocols---, Physical Layer Security (PLS) could also be exploited to improve the security of wireless channels \cite{shiu2011physical}. Most PLS schemes are based on the idea of leveraging intrinsic randomness in radio channel parameters \cite{maurer1999unconditionally,maurer1993secret,wilson2007channel,ye2007secrecy}, and are divided into key-less and key-based methods \cite{mukherjee2014principles}.
	\subsection{Key-less PLS}
	This type of secrecy does not require a key for encrypting data, yet it uses the channel properties of legitimate users and eavesdroppers to achieve secrecy \cite{mukherjee2014principles}. Indeed, this class of techniques, pioneered by Shannon \cite{shannon1949communication}, secures communication from an information-theoretic perspective. Shannon introduced a noiseless cipher system which is based on transmission limits subject to both reliability and secrecy. Later, Wyner considered noisy communication for the Shannon's cipher system, and introduced the so-called wire-tap channel \cite{wyner1975wire}. The secrecy provided by Wyner's method relies on random coding existence results. These random codes should have two features: (i) error-free detection for the legitimate receiver; and (ii) total randomness in erroneous detection, i.e., a bit error rate (BER) of 50\% for the Eve. In recent years, there have been codes (constructed from linear structures) with iterative decoding that come quite close to the Shannon limit in terms of providing error-free detection. However, these codes do not necessarily satisfy the second requirement of total randomness in erroneous detection. Random coding for a reliable and secure transmission to achieve the secure channel capacity has some fundamental differences with the case of channel capacity, and the design of channel codes that could mimic the random codes for secure and reliable transmission is in its infancy. 
	
	A common metric used in PLS to assess and compare different practical transmission schemes is the security gap \cite{klinc2011ldpc}. This refers to the gap in term of loss in energy to the optimum benchmark for reliability and secrecy obtained from the underlying information theoretic existence results (due to Wyner and relevant subsequent works) on the level of secrecy. Many recent practical coding schemes directly aim to reduce this security gap. In \cite{klinc2011ldpc}, the authors proposed to use punctured Low-Density Parity-Check (LDPC) codes for PLS. \cite{baldi2012coding} used scrambling, concatenation, and hybrid automatic repeat-request to LDPC and Bose–Chaudhuri–Hocquenghem (BCH) codes to farther lessen the security gap. Using polar and LDPC codes at the same time is also proposed in \cite{zhang2014polar}. Co-set codes constructed based on convolutional and turbo codes are proposed in \cite{nooraiepour2017randomized}. Also, the possibility of exploiting Serially Concatenated Low Density Generator Matrix (SCLDGM) codes in PLS schemes is considered in \cite{nooraiepour2018randomized}. An overview
	of error-control coding techniques for PLS can be found in \cite{aghdam2018overview}. It is seen that, in the best case, when a BER of $50\%$ is desired for the Eve, security gaps as small as $1~dB$ could be achieved by appropriately designing codes with long length \cite{aghdam2018overview}. 
	
	In fact, the simple-to-compute property of security gap makes it a useful criterion to avoid potential complexities associated with computing the original information-theoretic metrics. Nevertheless, there are no known techniques to quantify the impact of any such gap. This issue is non-existent in the technique proposed in this paper.
	
	As another approach in key-less PLS, the idea proposed by Goel and Nagi in \cite{negi2005secret, goel2005secret,goel2008guaranteeing} achieves the pefect secrecy, where additional artificially interfering signals are transmitted to the channel along with the original signals. In this method, the Eve's channel is degraded by injecting artificial noise into the null space of the legitimate user's channel. However, this technique sacrifices some power resources which consequently decreases the main channel capacity.  
	
	\subsection{Key-based PLS}
	On the contrary, key-based PLS schemes do not require an intricate computation, and they are practically more feasible \cite{ahlswede1993common,maurer1993secret}. Channel reciprocity is one of the main principles used for the process of generating key in wireless communications. This feature implies that the transmitted signal from both legitimate parties experiences almost the same fading, and thus the same key is generated at the legitimate nodes.
	Among different reciprocal features representing unpredictable radio channel parameters, Received Signal Strength (RSS) and Channel State Information (CSI) are the two common source of randomness. RSS has been the most favoured channel parameter used for key generation \cite{zhang2016key}. However, one of the main shortcomings of key generation methods based on RSS is that only one RSS value can be obtained during the channel coherence time, limiting the Key Generation Rate (KGR). On the other hand, CSI, which mainly refers to Channel Impulse Response (CIR) in time domain and Channel Frequency Response (CFR) in frequency domain, is proven to be suitable for key generation \cite{mathur2011proximate}. In these schemes, both amplitude and phase information of CSI could be exploited. It is worth noting that the channel phase, in comparison with channel amplitude, could potentially grant a better level of secrecy for the following reasons:
\begin{itemize}
    \item{The distribution of the phase, specifically in our case, will be uniform in $[0,2\pi)$ \cite{goldsmith2005wireless}. Therefore, any conditional density built by Eve based on its local observation of phase and/or magnitude conveys zero information about the phase values shared between legitimate nodes.} 
    \item{The magnitude has a known probability density function and Eve can form a conditional density function by conditioning on its local observations of signal magnitude, while accounting for factors such as the distance that affects the signal strength.}
    \item{The channel phase is more sensitive to the temporal variation of wireless channels. For instance, a small movement of transmitter, receiver or any other objects in the vicinity of legitimate users causes a noticeable change in the phase of the channel, while the amplitude might remain unchanged}.
\end{itemize}

Henceforth, phase is a promising random characteristic of radio channels for key generation schemes. The idea of using phase reciprocity for key generation was suggested in \cite{hershey1995unconventional}. In this work, both parties send $T$ known tones in both directions and by calculating the successive differences between the phase angles of the various tones, they establish $T-1$ shared phases. In \cite{hassan1996cryptographic}, each of the legitimate parties send two unmodulated tones with equal phases over two different frequencies; then, the parties quantize the phase difference between the two tones to establish a shared key. In continuation of \cite{hassan1996cryptographic}, authors in \cite{sayeed2008secure} find the probability that both ends of a link generate the same quantization index for a particular phase difference. Afterward, \cite{koorapaty2000secure} uses pre-coding for secure transmission schemes proposed in \cite{hassan1996cryptographic}. In \cite{koorapaty2000secure}, the first node sends M sinusoids with the same phase over M different frequencies; then, the second node calculates the $M-1$ phase differences between the received $M$ tones to pre-code Phase Shift Keying (PSK) symbols to be sent back to the first node; then, the first node calculates $M-1$ phase differences between the received $M$ tones which enables it to decode the information content of PSK symbols. Inspired by \cite{koorapaty2000secure}, 
in \cite{lai2017secure}, a secure transmission scheme for downlink Sparse Code Multiple Access (SCMA) by rotating the constellations using reciprocal channel phases is proposed. Furthermore, in \cite{althunibat2017physical}, the authors suggested to use reciprocal channel phases to rotate the transmitted signal, and to change the modulation type among a supported set. Although the pre-coding operation used in \cite{koorapaty2000secure,lai2017secure,althunibat2017physical} is similar to the masking operation performed in the current article, these works do not discuss the relevant information hiding capabilities. In general, the problem with all the above-mentioned methods (which exploit phase reciprocity to establish a shared key) is that if one legitimate user transmits a sequence of data in less than the coherence time of the channel, there will be a strong correlation between the subsequent transmissions which could be exploited by Eve. This problem is more severe in static or slow-fading channels (where the channel temporal variation rate is very slow), reducing the effective KGR.

To tackle the above-mentioned problems, some research works have proposed the idea of random beam-forming to enhance KGR when exploiting reciprocity over a static channel \cite{madiseh2012applying,cheng2015secret,li2017security}. The hardware complexities of realization of these methods are significantly high; because random beam-forming requires multiple active antennae, each requiring a dedicated base-band and Radio Frequency (RF) front-end. However, our approach relies on a single antenna and a single RF front end, and randomized switching of parasitic elements is performed using low-cost Positive-Intrinsic-Negative (PIN) diodes. 

In this article, a new practical system for achieving unconditional secrecy in key exchange based on phase reciprocity is presented. The main contribution is in practical implementation of the ideas first introduced in \cite{khandani2013two,khandani2017full}, and continued in \cite{khandani2019practical}. Indeed, the protocols of key sharing discussed in \cite{khandani2019practical} are enhanced to make them more suitable for practical implementation. Furthermore, mathematical proof for completely masking the content of PSK symbols, and for unconditional secrecy obtained in the proposed method is provided. 

Our method exploits the fact that the channel phase is uniformly distributed in $[0,2\pi)$. This allows completely masking (hiding) PSK symbols by modulo-$2\pi$ addition of phase values (analogous to modulo-2 addition in binary XOR). In a first step, a number of reciprocal phase values are measured at the two legitimate nodes connected through a wireless link. Due to channel reciprocity, the phase values measured at the two ends will be the same except for small deviations due to various sources of error, such as independent noise in legitimate parties' devices and time/frequency mismatch. Data bits to be communicated securely across the channel, after going through Forward Error Correction (FEC), are mapped to PSK constellation points, and each such PSK constellation is masked (rotated) using one of the shared phase values. An RF structure composed of a central transmit antenna surrounded by switchable parasitic RF elements is presented. This structure is used to perturb the RF environment of the transmitter and/or receiver antennas. This enables our method to offer a high KGR, and further, tackle with the problem of phase correlation in static and slow fading channels.

As mentioned earlier, it is difficult to translate the existent information-theoretical results which are based on random coding to practical implementation that guarantees both reliability and security (50$\%$ bit error for any erroneous transmissions).
One of the major benefits associated with the technique proposed here is that the mismatch between phases can be corrected by means of standard FEC techniques.

The rest of the paper is organized as follows. Section \ref{secrecy} provides the proof for the unconditional secrecy granted by masking PSK symbols. Section \ref{RF-mirrors} defines an RF mirror structure by which the channel in the vicinity of the legitimate nodes is perturbed, resulting in independent channel realizations. Section \ref{Configurations} elaborates on two different antenna configurations to establish common shared phase values between two legitimate parties. Section \ref{modulation} states a modulation technique for the proposed key exchange protocol. In Section \ref{results}, some experimental results are presented, and finally, Section \ref{conclusion} concludes the paper.

\emph{Notations:} In the following, entropy is denoted by $H(.)$ and the mutual information is denoted by $I(.;.)$. We use $u \perp v$ to show that $u$ and $v$ are independent. Notation $\oplus$ is used to represent modulo-$2\pi$ addition.
\section{Perfect secrecy of the proposed key generation method} \label{secrecy}
A binary key should be ideally composed of Independent and Identically Distributed  (i.i.d.) binary digits with the probability of zero and one being equal to $\frac{1}{2}$. This maximizes the entropy of the key for a given length of binary digits. This in turn renders eavesdropping as difficult as possible. Under such an assumption, a key of length $L$ will contain $L$ bits of information. On the other hand, to combat the effect of channel noise and other imperfections typical in wireless transmission, the $L$ key bits should undergo FEC, adding $r$ bits of channel coding redundancy to the original $L$ bits. It follows that $L+r=m N_p$, where $N_p$ is the number of $2^m$-PSK symbols carrying the key and the associated redundancy added by the deployed FEC. The original $L$ bits are selected to have the maximum possible information, i.e., $L$ bits of information. On the other hand, the $r$ bits of channel coding redundancy are dependent to the original $L$ bits, carrying no information. Therefore, the content of the encoded key stream of $L+r$ binary digits contains only $L$ bits of information. This brings up a question/concern if Eve can use the $r$ redundant bits to extract any information about the shared key. In the proposed protocol of key sharing, each $2^m$-PSK symbol is rotated by an independent shared phase value uniformly distributed in $[0,2\pi)$ obtained by the reciprocal property of the channel between the legitimate users. On the other hand, the following theorem proves that, if two phases are added modulo-$2\pi$, and if only one of the two phases is uniformly distributed in $[0,2\pi)$, the summation will not convey any information about the summands. Consequently, regardless of the FEC redundancy in the encoded stream of key bits, the addition of PSK constellations with (reciprocally) shared channel phases will (i) completely mask the content of each symbol, and (ii) make the received $N_p$ symbols independent. Therefore, FEC does not reveal any information to Eve, and its only negative side effect is that it decreases the key rate by a factor of $\frac{L}{L+r}$.

\begin{theorem} \label{uniform}
	Assume that angle $x$ is uniformly distributed in $[0, 2\pi)$, and  $y$ is another random angle, independent of $x$, distributed in  $[0, 2\pi)$ with certain Probability Distribution Function (PDF). Then modulo-$2\pi$ addition of $x$ and $y$ is also uniformly distributed in $[0, 2\pi)$.
\end{theorem}
\begin{proof}
	Denote by $f_{\alpha}(.)$ and $F_{\alpha}(.)$ the PDF and Cumulative Distribution Function (CDF) of the random variable $\alpha$, respectively. Also, $p\{E\}$ represents the probability of an event $E$.
	First, we compute $f_{X+Y}(t)$, and then, accordingly, compute $f_{X \oplus Y}(t)$. 
	
	Since $x$ and $y$ are two independent random variables, the PDF of their sum, i.e., $f_{X+Y}(t)$, is computed by convolving their respective PDF as follows  
	\begin{align} \label{conv}
		f_{X+Y}(t)=\int_{-\infty}^{\infty}	f_X(t-\tau)f_Y(\tau)d\tau.
	\end{align}
	
	As $x$ and $y$ are defined in $[0, 2\pi)$, $f_{X+Y}(t)$ has  non-zero values only in $[0, 4\pi)$. To compute~\eqref{conv}, the range of $t$ is partitioned into two parts, namely, $0 \leq t \leq 2\pi$ and $2\pi \leq t < 4\pi$, and then~\eqref{conv} is computed separately over each part. 
	Noting that $f_X(x)=\frac{1}{2\pi}$ ($x$ is uniformly distributed),
	for $0 \leq t \leq 2\pi$, equation \eqref{conv} becomes
	\begin{align} \label{conv1}
		f_{X+Y}(t)&=\int_{0}^{t} \frac{1}{2\pi} f_Y(\tau)d\tau
		=\frac{1}{2\pi} F_Y(t)
	\end{align} 
	and for $2\pi \leq t < 4\pi$, equation \eqref{conv} becomes
	\begin{align} \label{conv2}
		f_{X+Y}(t)&=\int_{t-2\pi}^{2\pi} \frac{1}{2\pi}  f_Y(\tau)d\tau  
		=\frac{1}{2\pi} \big( F_Y(2\pi) - F_Y(t-2\pi) \big)  
	=\frac{1}{2\pi} \big( 1 -  F_Y(t-2\pi) \big) 
	\end{align}
	and otherwise $f_{X+Y}(t)=0$. 
	
	In the following, we will compute $f_{X \oplus Y}(t)$ based on $f_{X+Y}(t)$ obtained above (note that to compute $f_{X \oplus Y}(t)$ we need to consider the modulo-$2\pi$ operation). 
	
	For a fixed value $0 \leq t_0 \leq 2\pi$, if $x \oplus y=t_0$, then either (i) $x + y=t_0$, or (ii) $x + y=t_0+2\pi$ (note that $0 \leq x+y < 4\pi$). Case (i) happens when $0 \leq x+y \leq 2\pi $ and case (ii) happens when $2\pi \leq x+y < 4\pi$. Therefore, from the Bayes' theorem, $p\{x \oplus y=t_0\}$ is computed as follows
	\begin{align} \label{px+y}
		p\{x \oplus y=t_0\} 
		&= p \{ 0 \leq x+y \leq 2\pi \} ~ p \{x+y=t_0 | 0 \leq x+y \leq 2\pi \} \nonumber \\
		&+ p \{ 2\pi \leq x+y < 4\pi \} ~ p \{x + y=t_0+2\pi | 2\pi \leq x+y < 4\pi \}.
	\end{align}
	
	Thus, $F_{X \oplus Y}(t)$ can be computed as follows
	\begin{align} \label{CDF_x+y}
		F_{X \oplus Y}(t)
		&= p \{ 0 \leq x+y \leq 2\pi \} ~ p \{0 \leq x+y \leq t |0 \leq x+y \leq 2\pi \} \nonumber \\
		&+ p \{ 2\pi \leq x+y < 4\pi \} ~ p \{ 2\pi \leq x+y \leq t+2\pi | 2\pi \leq x+y < 4\pi \}  \nonumber \\
		&= \big( \int_{0}^{2\pi} \frac{1}{2\pi} F_Y(\tau) d\tau \big)\big( \int_{0}^{t} \frac{1}{2\pi} F_Y(\tau) d\tau \big) ~~~~~~~~~~~~~\nonumber\\
		&+ \big( \int_{2\pi}^{4\pi} \frac{1}{2\pi} \big( 1 -  F_Y(\tau-2\pi) \big) d\tau  \big)\big( \int_{2\pi}^{t+2\pi} \frac{1}{2\pi} \big( 1 -  F_Y(\tau-2\pi) \big) d\tau \big)
	\end{align}
	where the last equation is obtained from \eqref{conv1} and \eqref{conv2}.
	Let us use $M_y(.)$ as the anti-derivative of $F_y(.)$. Noting that $f_{X \oplus Y}(t)=\frac{d}{dt} \big( F_{X \oplus Y}(t) \big)$, and that $\frac{d}{dt} \big( \int_{t_0}^{t} F_Y(\tau) d\tau \big)=F_Y(t)$ for a constant value of $t_0$ \cite{strangcalculus}, from \eqref{CDF_x+y} we have
	\begin{align}
		f_{X \oplus Y}(t) &=\frac{1}{2\pi} \big( M_y(2\pi)-M_y(0)\big) \big( \frac{1}{2\pi} F_{y}(t) \big) 
		+
		\frac{1}{2\pi} \big( \frac{2\pi-M_y(2\pi)+M(0)}{2\pi} \big) \big( \frac{1}{2\pi} F_y(t) \big)
		=\frac{1}{2\pi}.
	\end{align}
	
	Therefore, the Theorem \ref{uniform} is proved.
\end{proof}

\section{ RF-mirrors Structure} \label{RF-mirrors}
Perturbing the RF environment of Transmit/Receive (TX/RX) antenna is an essential part of the proposed system.  Fig.\,\ref{RFm} shows the RF structure used for this purpose. This structure was first introduced in \cite{khandani2013media} for Media-based Modulation. A number of switchable parasitic elements, hereafter referred as RF mirrors, is surrounding one or two transmitting antenna(s). Each RF wall (mirror) surrounding the center antenna(s), depending on whether it is in the OFF or ON state, passes the energy to its external surface to propagate outside the enclosure, or reflects the energy back (interior of the enclosure). In this manner, the energy bounces back and forth among different surrounding walls, and in doing so, creates a pseudo-random propagation pattern. Consequently, the state of each wall, as a result of its interaction with other walls, will have a multiplicative effect on the total number of generated antenna patterns. This feature causes an exponential growth for  reflection/propagation patterns in terms of the number of RF mirrors. This is similar to the effect of forming the image of a mirror in another mirror in the construction of a kaleidoscope. The key difference is that in our case, each surrounding mirror in the kaleidoscope can be selectively turned on (to reflect light back to the interior of the kaleidoscope enclosure), or turned off (to let the light rays leave the kaleidoscope enclosure in the corresponding area). We refer to the different antenna patterns created by on-off mirrors as transmitter states. 
The randomization phenomenon caused by pseudo-random selection of on-off mirrors is enhanced by multi-path propagation between the transmitting unit and its corresponding distant receiver. In other words, the outgoing RF signal corresponding to each pseudo-random antenna pattern will take various independent paths in reaching the distant receiver, resulting in a different complex gain for multi-path fading. Overall, the proposed structure creates a rich-scattering environment within the enclosure, which is further enriched through external multi-path propagation. 

It is noteworthy that, unlike the methods in PLS which rely on the magnitude of the RF signal, our technique relies only on the phase which is much more prone to changes while propagating in a rich scattering environment. As a result, the channel phase corresponding to different transmitter states will be uniformly distributed in $[0,2\pi)$, and also will be independent of each other. Fig.\,\ref{normconst} shows a normalized received constellation, where the constellation is subtracted by its empirical (complex) mean in order to eliminate its bias with respect to origin. The reason of subtracting the bias is that, due to limitations in laboratory (indoor) physical area, the two nodes exchanging RF signals for key establishment have been placed within few meters of each other, resulting in a strong (unwanted) Line-of-Sight (LOS) component. However, the LOS would be absent in an outdoor environment with larger separation between the transmitter and receiver units. Each point in Fig.\,\ref{normconst} corresponds to a random combination of RF mirrors switched to ON or OFF state. The red zone in the center of the constellation reflects the received signal when all RF mirrors are OFF. Also, Fig.\,\ref{AK-pic3} shows samples of antenna patterns corresponding to four different random selection of ON/OFF mirrors. Fig.\,\ref{AK-pic5} shows simulation results for sample in-door and out-door environments obtained by exporting several antenna patterns into a wireless EM propagation software (REMCOM wireless InSite).
\begin{figure}[h!] 
    \centering
	\includegraphics[width=0.5\columnwidth]{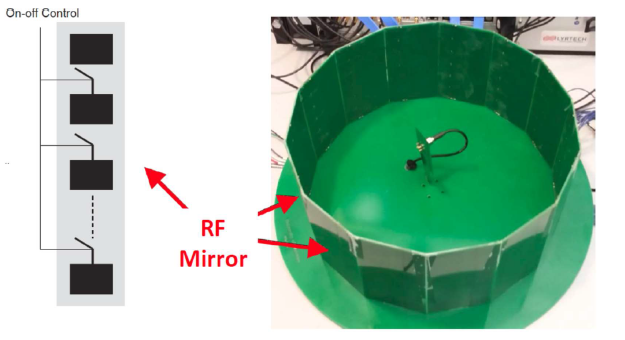}
	\caption{Structure of RF-mirrors.}
	\label{RFm}
\end{figure}

\begin{figure}[t] 
  \centering
  \subfloat[]{\includegraphics[width=5.5cm,keepaspectratio=ture]{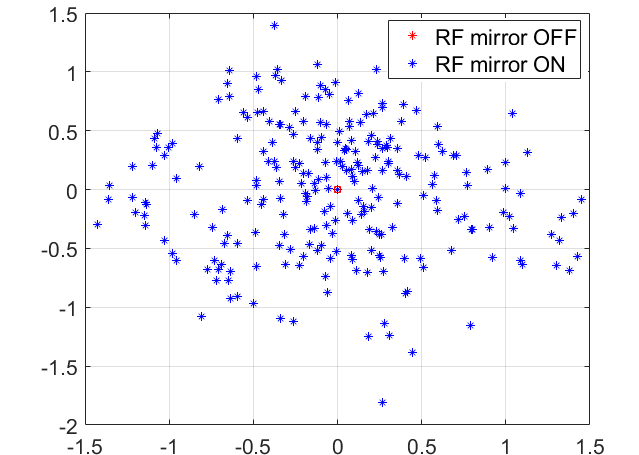}\label{normconst}}
  \subfloat[]{\includegraphics[width=5cm,keepaspectratio=ture]{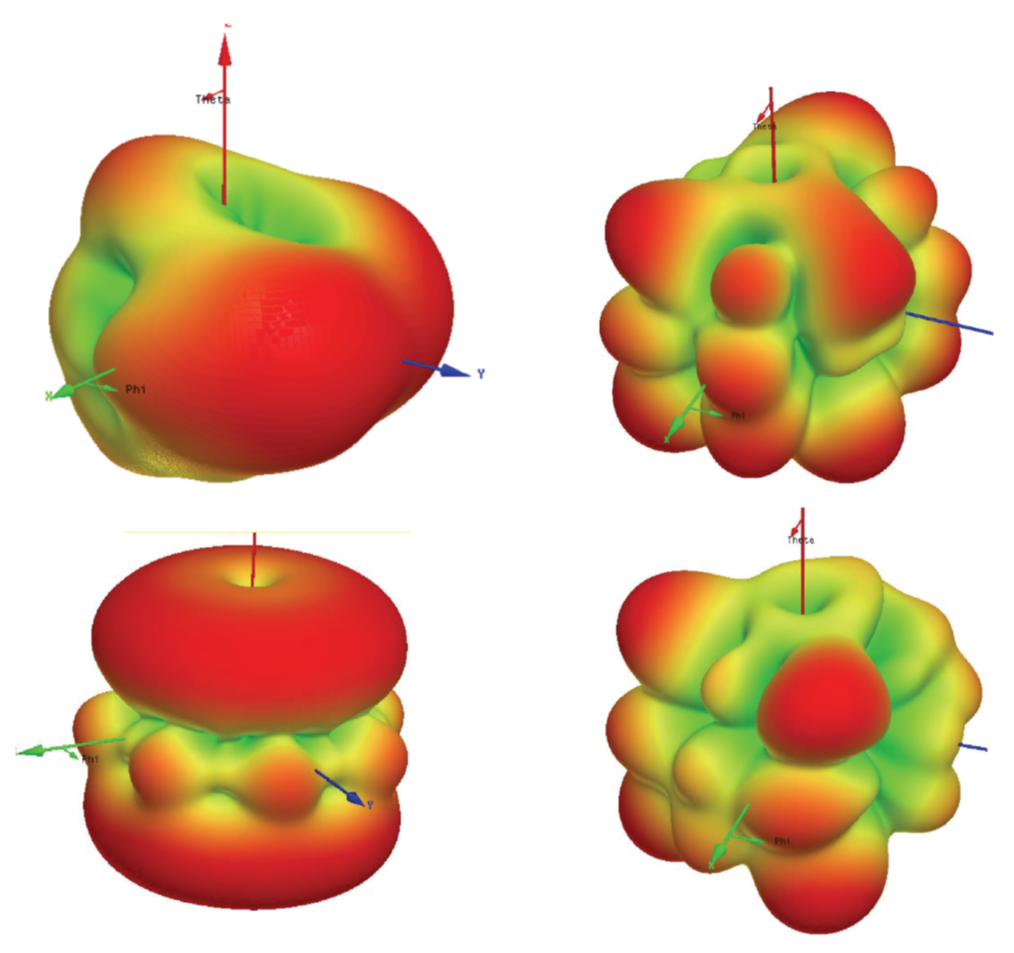}\label{AK-pic3}}\\
  \caption{Randomness caused by the mirror structure. (a) Normalized constellation obtained from the over-the-air transmission. (b) Samples of antenna patterns corresponding to four random states of ON/OFF mirrors.}
  \label{Fig:model_verification}
  \vspace{-0.5cm}
\end{figure}

\begin{figure}[h!] 
	\includegraphics[width=.5\textwidth,center]{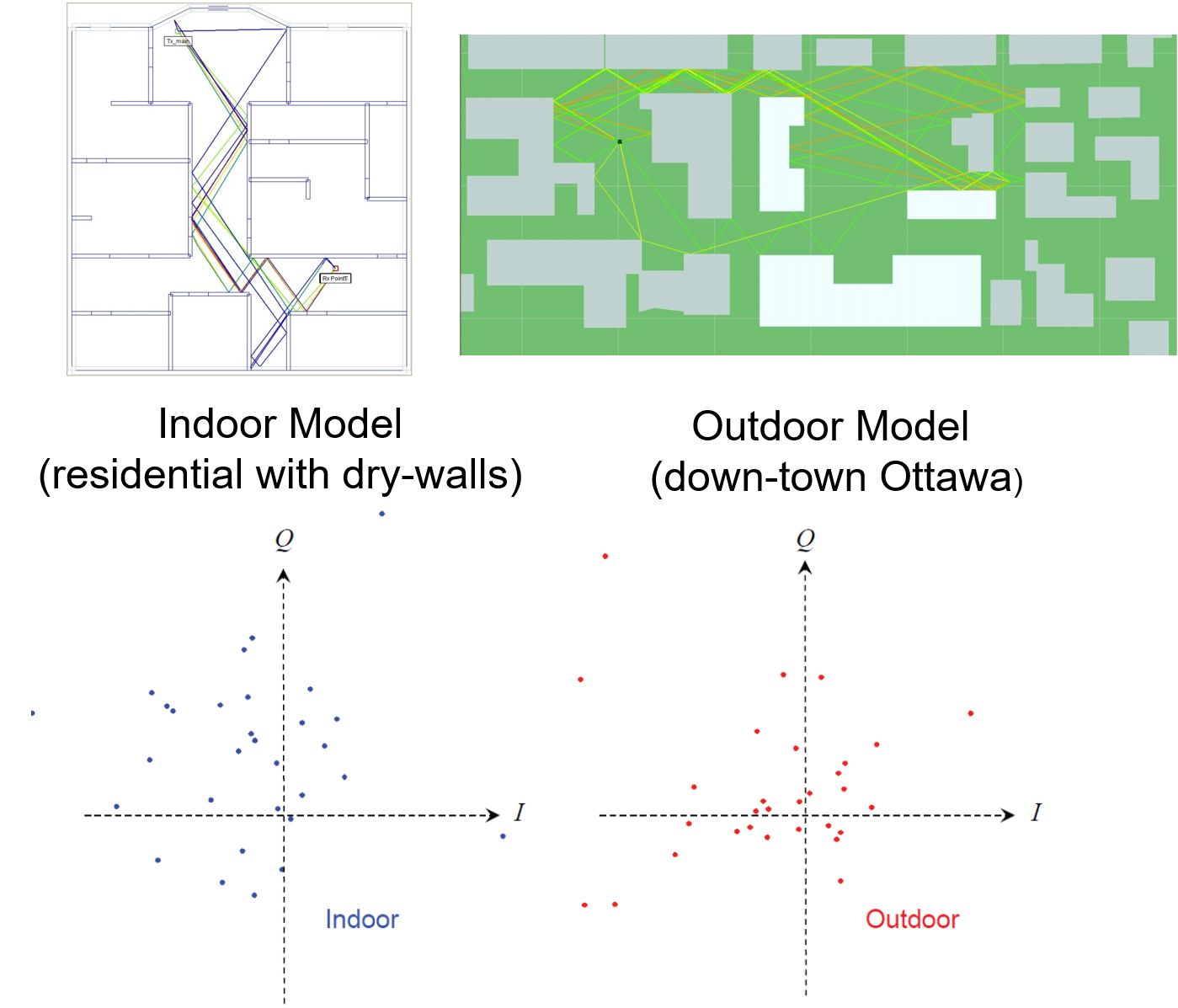}
	\caption{Simulation results for sample in-door and out-door environments obtained by exporting several antenna patterns into wireless EM propagation software REMCOM wireless InSite.}
	\label{AK-pic5}
\end{figure}

\section{Antenna Configurations} \label{Configurations}
The main condition for perfect secrecy is that each transmit antenna should transmit only once in each channel state, and then its associated RF mirrors should be changed to a new configuration for the next transmission (to establish an independent phase value).

In this section, we present two different configurations satisfying these requirements. In addition, we show that a potential eavesdropper cannot attain any information when two legitimate nodes are exchanging keys. Similar to assumptions in other papers \cite{mathur2008radio}, it is assumed that the eavesdropper is aware of the key exchange protocols. 

The first configuration uses a single antenna at each legitimate node (two-antenna system), and the second configuration uses two antennas at each legitimate node (four-antenna system). As it will be discussed in this section, the four-antenna system resolves the issue of synchronization associated with two-antenna system, yet it has a higher hardware complexity. 

\subsection{Two-antenna System} \label{Sec:two}
In the two-antenna configuration, the reciprocal phase is measured as the phase of the channels from Alice to Bob and from Bob to Alice, represented as Alice$\to$Bob and  Bob$\to$Alice, respectively. This protocol of key sharing relies on using a pilot which is known to both Alice and Bob. Also, the RF signal is modulated using Orthogonal Frequency-Division Multiplexing (OFDM) over $S$ frequency tones, i.e., $\{f_1,f_2,...,f_S\}$.  

\subsubsection{Notations and Assumptions} 
Assume that only Alice's transmitter is equipped with the RF mirror structure explained in Section \ref{RF-mirrors}. Note that although having the mirror structure at both Alice's and Bob's sides enhances the richness of the state space, it is adequate to have only one side equipped with the mirror structure.
Denote by $K$, the number of mirrors at Alice's side.  Therefore, based on whether a mirror is in the ON or OFF state, $2^K$ different states will be realized. We use $\alpha^i$ to show that the phase $\alpha$ is measured when the mirrors are in state $i$, $1 \leq i \leq 2^K$.
We assume that the Eve is equipped with $n$ antennas overhearing the messages transmitted between Alice and Bob. As indicated in Fig.\,\ref{fig:Eve-2-antennas}, Eve's antennas are denoted by $\{\mathsf{e}_1,\mathsf{e}_2,...,\mathsf{e}_n\}$. We denote Alice's and Bob's antennas by $\mathsf{a}$ and $\mathsf{b}$, respectively. When the Alice's mirrors are in state $i$, we denote by $\mathsf{ab}^i$ the phase that Bob records when Alice acts as the transmitter; and by $\mathsf{ba}^i$ the phase that Alice records when Bob acts as the transmitter. Furthermore, we indicate by $\mathsf{ae}_k^i$ the phase from Alice's antenna to the $k^{th}$ antenna of Eve when the mirrors are in state $i$. Note that as the mirrors are on the Alice's side, Eve's observations from Bob's antenna are independent of the mirror state. Hence, we denote by $\mathsf{be}_k$ the phase angle from Bob's antenna to the Eve's $k^{th}$ antenna (see Fig.\,\ref{fig:Eve-2-antennas}).  

\begin{figure}[h]
	\centering
	\includegraphics[width=0.45\textwidth]{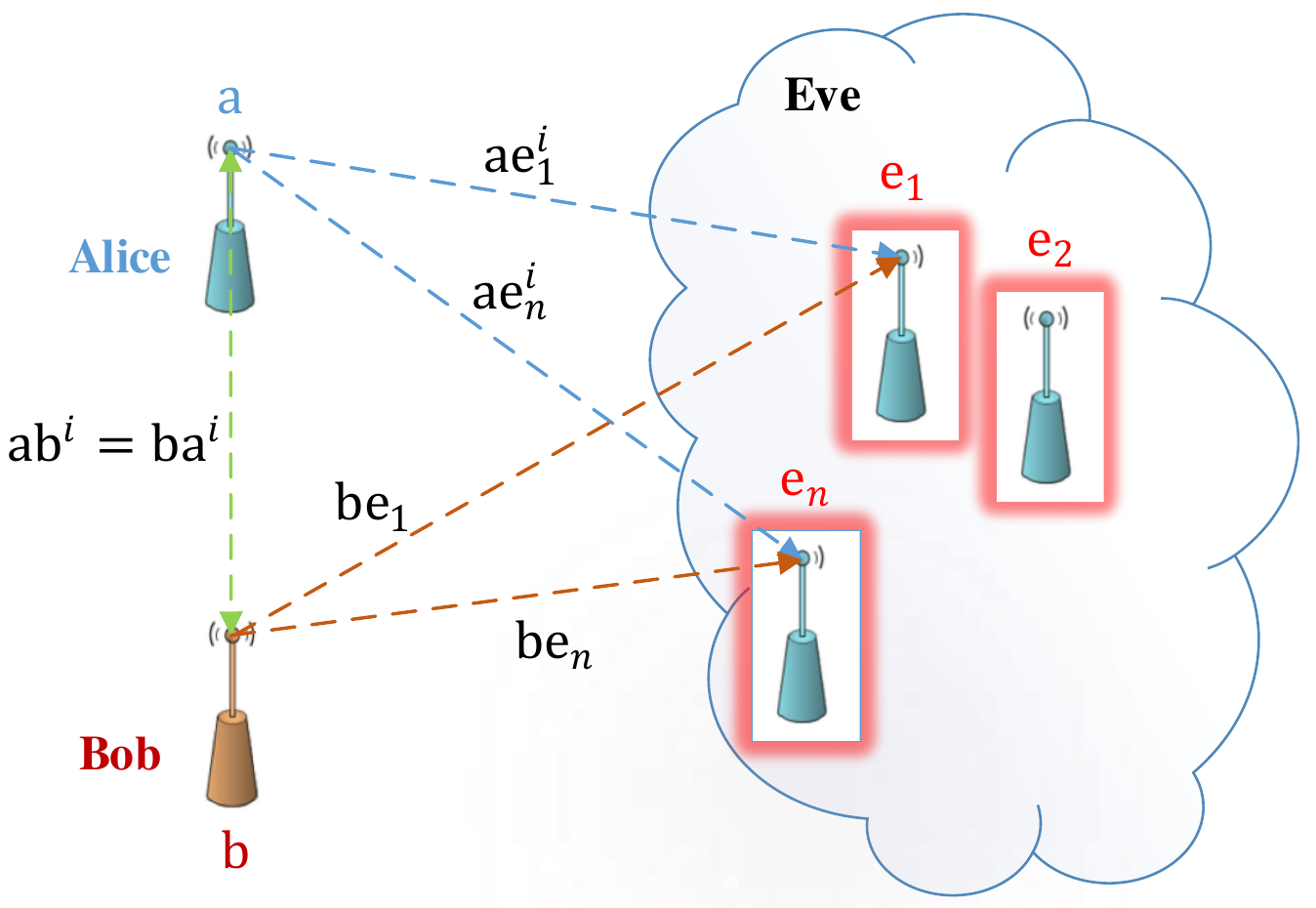}
	\caption{The two-antenna system. Eve's observations from both Alice's and Bob's antennas are indicated when the Alice's mirrors are at the state $i$. Eve's observations from Bob's antenna are independent of the mirror's state as the mirror structure is on the Alice's side.}
	\label{fig:Eve-2-antennas}
\end{figure}

\subsubsection{The protocol of Key sharing} 
To attain a shared key, each cycle of key sharing is composed of four transmissions: two consecutive transmissions from Alice to Bob, followed by two consecutive transmissions from Bob to Alice. In the following, one cycle of transmission is elaborated by referring to Fig.\,\ref{Preamble}.

\begin{figure}
    \centering
	\includegraphics[width=\textwidth,height=4.5cm]{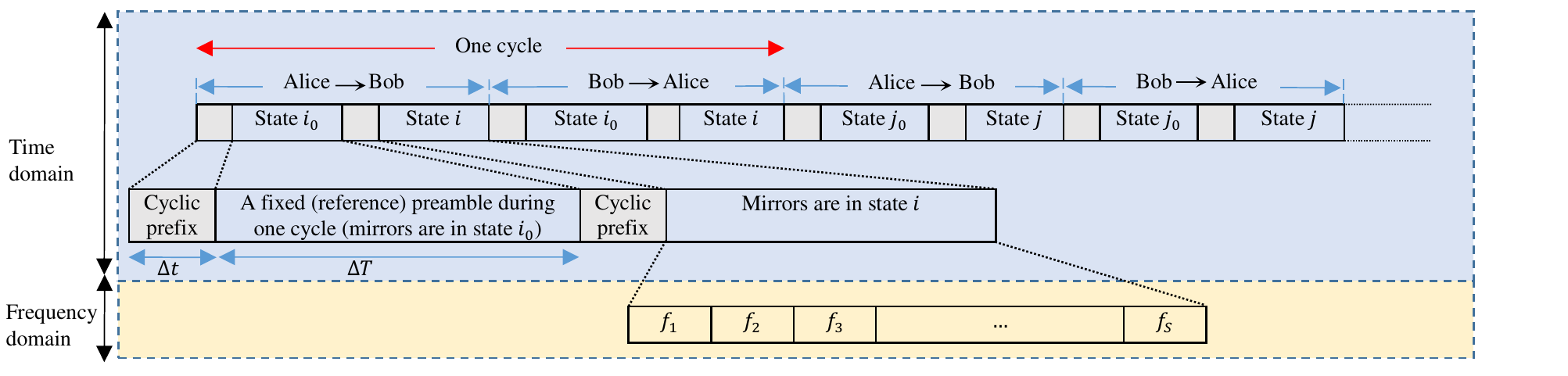}
	\caption{The protocol of key sharing in the two-antenna system. The time slot corresponding to each state in the time domain is divided into $S$ tones in the frequency domain.}
	\label{Preamble}
\end{figure}

\begin{itemize}
	\item{\textit{First transmission}:}
	Alice sends a preamble (pilot) to Bob using OFDM modulation. Then, Bob measures the phase of the received signal over each tone, i.e., Bob records $\mathsf{ab}^{i_0}(f_l)$ for $1 \leq l \leq S$; note that $\mathsf{ab}^{i_0}$ is a function of frequency.
	
	\item{\textit{Second transmission}:}
	In this phase, the mirrors must be switched to a new state $i$. To avoid intersymbol interference (ISI), a cyclic prefix with duration of $\Delta t$ is used. Note that changing the state of the mirrors increases the delay spread of the channel, and therefore switching to a new state must take place in the early parts of the cyclic prefix to ensure that $\Delta t>\tau_{max}$, where $\tau_{max}$ is the maximum delay spread of the channel.
	After the cyclic prefix duration is over, Alice sends another RF signal to Bob using OFDM modulation. Bob measures the change in the phase of the pilot in each frequency band with respect to the corresponding tone in the preamble, i.e., Bob records $\big(\mathsf{ab}^i(f_l) \big) \ominus \big(\mathsf{ab}^{i_0}(f_l)\big)$ for $1 \leq l \leq S$. 
	
	\item{\textit{Third transmission}:}
	Now, the roles of Alice and Bob are reversed. Bob sends the same preamble to Alice as the one that Alice sent to Bob in the \textit{First transmission}. Therefore, in this protocol, it is required to have a preamble known to both legitimate nodes. To obtain the same preamble, the Alice's mirrors are returned back to the reference states $i_{0}$, and after the duration of cyclic prefix $\Delta t$, an RF signal is sent from Bob to Alice. Alice measures the phase of the received signal, i.e., Alice records $\mathsf{ba}^{i_0}(f_l)$ for $1 \leq l \leq S$. 
	
	\item{\textit{Forth transmission}:}
	Finally, Alice's mirrors state is changed to $i$, and after the duration of cyclic prefix (which is a part of OFDM signal), another RF signal will be sent from Bob to Alice. Alice measures the change in the phase of the pilots in each frequency band with respect to the corresponding tone in the preamble, i.e., Alice records $\big(\mathsf{ba}^i(f_l) \big) \ominus \big(\mathsf{ba}^{i_0}(f_l)\big)$ for $1 \leq l \leq S$. 
\end{itemize}
\begin{claim}
Based on the cycle of transmission elaborated above, Alice and Bob can obtain a shared common phase value as $\big(\mathsf{ab}^i(f_l) \ominus \mathsf{ab}^{i_0}(f_l) \big)=\big( \mathsf{ba}^i(f_l) \ominus \mathsf{ba}^{i_0}(f_l) \big)$.
\end{claim}
\begin{proof}
First note that the change in the phase of an RF signal from the transmitter to receiver is affected by the signal traversing through: (i) the RF channel between the nodes, and (ii) both transmit and receive chains.
For Alice and Bob to obtain a shared phase value based on their recorded phases in the above four transmissions, the effect of phases induced by receive/transmit chains must be taken into account.
To shed more light, when Alice is the transmitter, the transmit chain of Alice and the receive chain of Bob contribute to the phase change of the RF signal. On the other hand, when Bob is the transmitter, the transmit chain of Bob and the receive chain of Alice contribute to the phase change. In the following, we will elaborate on how the phases induced by receive/transmit chains are mitigated in our protocol.

When Bob calculates $\big(\mathsf{ab}^i(f_l) \big) \ominus \big(\mathsf{ab}^{i_0}(f_l)\big)$ in the \textit{Second transmission}, the phase changes induced by Alice's transmit chain and Bob's receive chain are cancelled out. Similarly, when Alice calculates $\big(\mathsf{ba}^i(f_l) \big) \ominus \big(\mathsf{ba}^{i_0}(f_l)\big)$ in the \textit{Forth transmission}, the phase changes induced by Bob's transmit chain and Alice's receive chain are cancelled out. On the other hand, the reciprocity of the channel necessitates that $\mathsf{ab}^i(f_l)= \mathsf{ba}^i(f_l)$ and $\mathsf{ab}^{i_0}(f_l)= \mathsf{ba}^{i_0}(f_l)$; therefore $\big(\mathsf{ab}^i(f_l) \big) \ominus \big(\mathsf{ab}^{i_0}(f_l)\big)=\big(\mathsf{ba}^i(f_l) \big) \ominus \big(\mathsf{ba}^{i_0}(f_l)\big)$ which could be used as a common shared phase value (note that as the time interval between two OFDM symbols is very small, the hardware physical conditions, such as its temperature, do not notably change. Thus, the hardware behaviour---in terms of the phase change it exerts over the signal---remains the same between two consecutive transmissions).
\end{proof}

In the following, to suppress the effect of Additive White Gaussian Noise (AWGN), we propose obtaining only one common phase value from the $S$ phases shared in the above manner (in lieu of using all of the $S$ common phases as shared values).   

First, at the transmit side, the pilots' amplitudes are selected to be $\pm 1$ (there are $S$ pilots in $S$ frequency bands). Then, at the receive side, by applying appropriate sign changes, the receiver coherently adds up the pilots in the $S$ frequency bands, and then computes their average. Then, the phase of the averaged pilots will be used as the only shared phase value. In other words, Bob computes $\mathsf{ab}^i$ and  $\mathsf{ab}^{i_0}$ as the phase of the averaged pilots in the \textit{First} and \textit{Second transmissions}, respectively. Similarly, Alice finds $\mathsf{ba}^i$ and  $\mathsf{ba}^{i_0}$ as the phase of the averaged pilots in the \textit{Third} and \textit{Forth transmissions}, respectively. Thereafter, Bob and Alice use $ \left( \mathsf{ab}^i \ominus \mathsf{ab}^{i_0} \right) $ and $\left( \mathsf{ba}^i \ominus \mathsf{ba}^{i_0}\right)$ as a shared phase value to establish a common key. 
\begin{remark}
	Instead of simple averaging at the receiver, a more complex coding technique could be applied over the pilots received in $S$ frequency bands. Nevertheless, averaging the pilots can be regarded as a repetition code in frequency domain. Note that spectral efficiency is a secondary concern in this paper because a key can be used over many transmissions before it is changed. 
\end{remark}
\begin{remark}
	As discussed earlier, the length of the cyclic prefix $\Delta t$ is determined by the delay spread of the channel. Denote by $\Delta T$ the time length during which one node sends an RF signal to the other node, as depicted in Fig.\,\ref{fig:Eve-2-antennas}. Assuming a fixed available bandwidth, if the number of frequency bands $S$ decreases, $\Delta T$ also decreases. This, in turn, decreases the efficiency of the OFDM, as for a fixed-length cyclic prefix, $\Delta T$ has decreased. On the other hand, if $S$ increases, $\Delta T$ also increases (assuming a fixed available bandwidth) yielding a higher efficiency. Nevertheless, the number of sub-bands could be increased to the extent that the time interval during which one node completes its two consecutive transmissions, i.e., $2 \left(\Delta t +\Delta T \right)$, falls below the coherence time of the channel. Therefore, the number of sub-bands $S$ should be cautiously selected. 
\end{remark}

\subsubsection{Analysis} 
Next, we show that a potential eavesdropper equipped with $n$ receive antennas cannot attain the secret key shared by these two legitimate nodes. 

We assume Eve is aware of the protocol of key sharing used by the legitimate nodes, and therefore she finds the average of her observations over the utilized $S$ frequency bands. The observations made by Eve upon one cycle of key sharing between Alice and Bob are as follows:
\begin{itemize}
	\item{$\mathsf{a}$ $\to$ $\mathsf{b}$, when mirrors are at state $i_{0}$: Eve measures $\mathsf{ae}^{i_0}_{k}$},
	\item{$\mathsf{a}$ $\to$ $\mathsf{b}$, when mirrors are at state $i$: Eve measures $\mathsf{ae}^{i}_{k}$},
	\item{$\mathsf{b}$ $\to$ $\mathsf{a}$, when mirrors are at state $i_{0}$: Eve measures $\mathsf{be}_{k}$},
	\item{$\mathsf{b}$ $\to$ $\mathsf{a}$, when mirrors are at state $i$: Eve measures $\mathsf{be}_{k}$},
\end{itemize}
for $1 \leq k \leq n$.
In the following paragraphs, we show that based on these observations, Eve cannot extract any information about the shared phase, i.e., $\mathsf{ab}^i \ominus \mathsf{ab}^{i_0}$. First, we present a Lemma.
\begin{lemma} 
	$I \left( \mathsf{ae}_{k}^{i};\left( \mathsf{ab}^{i} \ominus \mathsf{ab}^{i_0} \right)| \mathsf{ae}_{k}^{i_0} \right)=0$.
	\label{I_lemma}
\end{lemma} 
\begin{proof}
	By the definition of conditional mutual information \cite{cover1999elements}, we have 
	\begin{align} \label{I1}
		I \left( \mathsf{ae}_{k}^{i};\left( \mathsf{ab}^{i_0} \ominus \mathsf{ab}^{i} \right)| \mathsf{ae}_{k}^{i_0} \right)
		=E_{p\left( \mathsf{ae}_{k}^{i},\left(\mathsf{ab}^{{i_0}} \ominus \mathsf{ab}^{j}\right)| \mathsf{ae}_{k}^{{i_0}}\right)}\left( \log \frac{p\left(\mathsf{ae}_{k}^{i},\left(\mathsf{ab}^{i_0} \ominus \mathsf{ab}^{i}\right)| \mathsf{ae}_{k}^{i_0} \right)}{p\left(\mathsf{ae}_{k}^{i}|\mathsf{ae}_{k}^{i_0}\right)p\left(\left( \mathsf{ab}^{i_0} \ominus \mathsf{ab}^{i}\right) |\mathsf{ae}_{k}^{i_0}\right)}\right)
	\end{align}
	where $E_{p\left(.\right)}\{.\}$ represents  the expected value over the probability space of $p\left(.\right)$. On the other hand, by Bayes' formula, we have
	\begin{align} \label{bayes'}
		p\left(\mathsf{ae}_{k}^{i},\left(\mathsf{ab}^{i_0} \ominus \mathsf{ab}^{i}\right)| \mathsf{ae}_{k}^{i_0} \right)
		=\frac{p\left(\mathsf{ae}_{k}^{i}|\left(\mathsf{ab}^{i_0} \ominus \mathsf{ab}^{i}\right),\mathsf{ae}_{k}^{i_0} \right)p\left( \mathsf{ae}_{k}^{i_0}, \left(\mathsf{ab}^{i_0} \ominus \mathsf{ab}^{i}\right)\right)}{p\left( \mathsf{ae}_{k}^{i_0} \right)}.
	\end{align}
	To simplify the terms in \eqref{bayes'}, it must be noted that by changing the state of the mirrors at the Alice's side, the RF environment in the vicinity of the Alice antenna $\mathsf{a}$ is changed (in effect, it changes the antenna pattern). As mentioned earlier, such a change in the vicinity of the transmit antenna changes the propagation path to the receiver antenna, further enhancing the randomness in the end-to-end channel phase. In other words, the randomness caused by RF mirrors will be augmented when the signal propagates to the destination and thereby interacts with randomly located objects in its propagation path. Interactions can be in the form of partial absorption/reflection which causes a phase shift. Numerous such signal paths, each experiencing a random phase shift, will add up at the destination antenna. In other words, the environment external to the transmit antenna is typically a rich scattering environment which enhances the randomness initially caused by the change in the Alice's antenna pattern using RF mirrors.
	Therefore, the phases that Eve records from Alice's antenna when the mirrors are at state $i$, i.e., $\mathsf{ae}_{k}^{i}$, are independent from the those she records when mirrors are in state ${i_0}$, i.e., $\mathsf{ae}_{k}^{i_0}$, meaning that
	\begin{align}
		\mathsf{ae}_{k}^{i} \perp \mathsf{ae}_{k}^{i_0} \label{ind1}.
	\end{align}
	
	Furthermore, it is obvious that as the location of Eve and Bob is different, $\mathsf{ae}_{k}^{i} \perp\mathsf{ab}^{i}$, and also $\mathsf{ae}_{k}^{i} \perp \mathsf{ab}^{i_0}$, which yields 
	\begin{align} 
		\mathsf{ae}_{k}^{i} \perp \left( \mathsf{ab}^{i_0} \ominus \mathsf{ab}^{i}\right) \label{ind2}.
	\end{align}
	
	In addition,  $\mathsf{ae}_{k}^{i_0} \perp\mathsf{ab}^{i_0}$ and $\mathsf{ae}_{k}^{i_0} \perp \mathsf{ab}^{i}$, which results
	\begin{align} 
		\mathsf{ae}_{k}^{i_0} \perp \left( \mathsf{ab}^{i_0} \ominus \mathsf{ab}^{i}\right) \label{ind3}.
	\end{align}
	
	From equations \eqref{ind1} and \eqref{ind2}
	\begin{align}
		&p\left(\mathsf{ae}_{k}^{i}|\left(\mathsf{ab}^{i_0} \ominus \mathsf{ab}^{i}\right),\mathsf{ae}_{k}^{i_0} \right)=p\left(\mathsf{ae}_{k}^{i} \right)\label{ind4}
	\end{align}
	and from \eqref{ind3} it is concluded 
	\begin{align}
		&p\left( \mathsf{ae}_{k}^{i_0}, (\mathsf{ab}^{i_0} \ominus \mathsf{ab}^{i})\right)=p\left( \mathsf{ae}_{k}^{i_0}\right)p\left( \mathsf{ab}^{i_0} \ominus \mathsf{ab}^{i}\right) \label{ind5}.
	\end{align}
	Substituting equations \eqref{ind4} and \eqref{ind5} in \eqref{bayes'}, results in
	\begin{align}
		p\left(\mathsf{ae}_{k}^{i},\left(\mathsf{ab}^{i_0} \ominus \mathsf{ab}^{i}\right)| \mathsf{ae}_{k}^{i_0} \right)=p\left(\mathsf{ae}_{k}^{i} \right) p\left( \mathsf{ab}^{i_0} \ominus \mathsf{ab}^{i} \right).
	\end{align}
	Therefore, equation \eqref{I1} is simplified to 
	\begin{align} \label{i=0}
		 I \left(\mathsf{ae}_{k}^{i};\left(\mathsf{ab}^{i_0} \ominus \mathsf{ab}^{i}\right)| \mathsf{ae}_{k}^{i_0} \right) 
		&=E_{p\left( \mathsf{ae}_{k}^{i},\left(\mathsf{ab}^{i_0} \ominus \mathsf{ab}^{i}\right)| \mathsf{ae}_{k}^{i_0}\right)} \left( \log \frac{p\left(\mathsf{ae}_{k}^{i} \right) p\left(\mathsf{ab}^{i_0} \ominus \mathsf{ab}^{i} \right)}{p\left(\mathsf{ae}_{k}^{i} \right) p\left( \mathsf{ab}^{i_0} \ominus \mathsf{ab}^{i} \right)}\right) \nonumber \\
		&=E_{p( \mathsf{ae}_{k}^{i},(\mathsf{ab}^{i_0} \ominus \mathsf{ab}^{i})| \mathsf{ae}_{k}^{i_0})} \log \left( 1\right)=0.
	\end{align}
\end{proof}
Note that, Lemma \ref{I_lemma} is focused on a single antenna of Eve, namely its $k^{th}$ antenna. To find the information captured by the other Eve's antennas, if $\mathsf{ae}_{k}^{i_0}$ is replaced by $\mathsf{ae}_{l}^{i_0}$ for $l\neq k $ in Lemma \ref{I_lemma}, using the same reasoning as that discussed above, it follows that all the equations \eqref{I1} to \eqref{i=0} are still correct. Therefore,
\begin{align} \label{ineqj}
	I \left(\mathsf{ae}_{k}^{i};\left(\mathsf{ab}^{i_0} \ominus \mathsf{ab}^{i}\right)| \mathsf{ae}_{l}^{i_0} \right)=0.
\end{align}
Since $k$ and $l$ are arbitrary in \eqref{ineqj}, we have
\begin{align} \label{i_tuple}
	I \left([\mathbf{ae}^{i}];\left(\mathsf{ab}^{i_0} \ominus \mathsf{ab}^{i}\right)| [\mathbf{ae}^{i_0}] \right)=0
\end{align}
where $[\mathbf{ae}^{i}]=\{\mathsf{ae}_{1}^{i},\mathsf{ae}_{2}^{i},...,\mathsf{ae}_{n}^{i}\}$, and similarly $[\mathbf{ae}^{i_0}]=\{\mathsf{ae}_{1}^{i_0},\mathsf{ae}_{2}^{i_0},...,\mathsf{ae}_{n}^{i_0}\}$ are the vectors of phases that Eve measures on all of its antennas from the Alice's antenna when the mirrors are in states $i$ and $i_0$, respectively.

Thus far, we have considered the Eve's observations from Alice's antenna, but not those from Bob's. In the following we inspect the Eve's observations from Bob's antenna. 
Using the same method and justification as those used in proving Lemma \ref{I_lemma}, it could be proven that $I \left(\mathsf{be}_{k};\left(\mathsf{ab}^{i} \ominus \mathsf{ab}^{i_0}\right)| \mathsf{be}_{k}\right)=0$. Furthermore, with the same reasoning as that explained in proving \eqref{i_tuple}, it is straight-forward to prove that 
\begin{align} \label{i_tuple2}
	I \left([\mathbf{be}];\left(\mathsf{ab}^{i_0} \ominus \mathsf{ab}^{i}\right)| [\mathbf{be}] \right)=0
\end{align}
where $[\mathbf{be}]=\{\mathsf{be}_{1},\mathsf{be}_{2},...,\mathsf{be}_{n}\}$ is the vector of phases that Eve receives on all of its antennas from $\mathsf{b}$ (Bob's antenna), no matter what the state of the mirrors is. 

Furthermore, since $\mathsf{ae}_{k}^{i} \perp \mathsf{be}_{k}$, all the equations \eqref{I1} to \eqref{i=0} are still correct. Thus, 
\begin{align} \label{i_tuple3}
	I \left([\mathbf{ae}^{i}];\left(\mathsf{ab}^{i_0} \ominus \mathsf{ab}^{i}\right)| [\mathbf{be}] \right)=0
\end{align}
and similarly,
\begin{align} \label{i_tuple4}
	I \left([\mathbf{be}];\left(\mathsf{ab}^{i_0} \ominus \mathsf{ab}^{i}\right)| [\mathbf{ae}^{i_0}] \right)=0.
\end{align}

From equations \eqref{i_tuple} to \eqref{i_tuple4}, it is obtained that
\begin{align} \label{final_I}
	I \left([\mathbf{ae}^{i},\mathbf{be}];\left(\mathsf{ab}^{i_0} \ominus \mathsf{ab}^{i}\right)|[\mathbf{ae}^{i_0},\mathbf{be}] \right)=0
\end{align}
where $[\mathbf{ae}^{i},\mathbf{be}]$ and $[\mathbf{ae}^{i_0},\mathbf{be}]$ are all the phases that Eve receives from both $a$ and $b$, i.e., Alice's and Bob's antennas, when the mirrors are at state $i$ and $i_0$, respectively.
Equation \eqref{final_I} states that even if Eve records all the phases received from the legitimate nodes during the four transmissions of a full cycle, still she cannot attain any information about the shared key.

\subsection{Four-antenna System} \label{Sec:four}
While having a simple structure, the downside of two-antenna system is the issue of synchronization between the legitimate nodes in both time and frequency domains. Time mismatch causes the receiver to inaccurately sample the received OFDM symbols, diminishing the effect of cyclic prefix. On the other hand, unmatched frequencies between the local oscillators of the transmitter and receiver results in inter-carrier-interference.  

To tackle this issue, four-antenna system is introduced in this section in which each of the two legitimate nodes is equipped with two antennas. As will be elaborated in the following, in such a system, the change in the phase of the signal is measured between the two antennas located on the same unit, and consequently, the problem of time and frequency mismatch is not of any concern. 


\subsubsection{Notations and Assumptions} 
We use the same notations and assumptions as those used for two-antenna system. However, as each node has an additional antenna with respect to the two-antenna system, we indicate the new antennas on Alice's and Bob's sides by $\alpha$ and $\beta$, respectively. In other words, we denote by $\{\mathsf{a},\mathsf{\alpha}\}$ and $\{\mathsf{b},\mathsf{\beta}\}$ the antennas for Alice and Bob, respectively. Note that similar to the two-antenna system, the RF mirror structure is only utilized on only of the two sides, say on Alice's side. The system structure along with the channels between the legitimate nodes and Eve are depicted in Fig.\,\ref{fig:Eve-4-antenna}.

Furthermore, similar to the two-antenna system, the signals are transmitted using OFDM over $S$ frequency bands, and by means of averaging, the effect of the noise will be reduced. For the sake of notational simplicity, although the averaging over $S$ frequency bands takes place, the process of averaging is not expressed explicitly. 
\begin{figure}[h!] 
	\includegraphics[width=.5\textwidth,center]{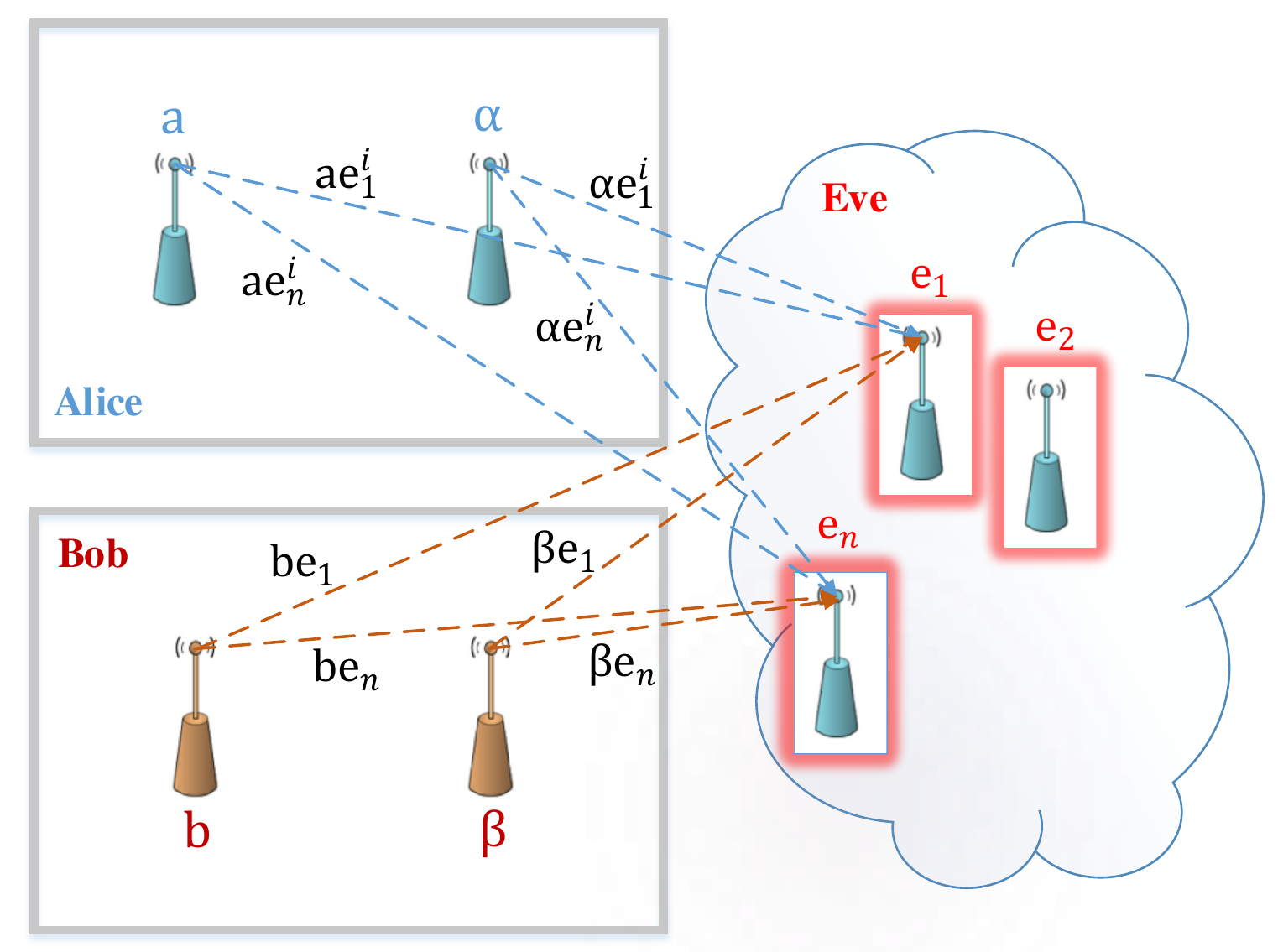}
	\caption{Eve's observations in four-antenna system when the Alice's mirrors are at the state $i$. Eve's observations from Bob's antennas are independent of the mirror's state.}
	\label{fig:Eve-4-antenna}
\end{figure}
\subsubsection{The protocol of Key sharing} 
As indicated in Fig.\,\ref{fig:4-antenna-config}, two traversing loops, one initiated by Alice and the other one initiated by Bob, will be completed to share a common phase value. 

It should be pointed out that exchanging the signal between the two antennas on the same side (i.e., transmissions between $\mathsf{a}$ and $\alpha$, or between $\mathsf{b}$ and $\beta$) takes place through the RF front-end (wired connection). The propagation of RF signal within the RF front-end of each unit results in a phase shift which is virtually constant over time. Note that variations in such interior-to-hardware phase shifts are caused by effects such as changes in temperature or aging of components which occurs very slowly over time. Consequently, the net effect is a constant phase difference between the phase shifts measured at the two legitimate nodes. Such a constant phase difference can be compensated through an initial calibration. An alternative is to estimate the constant phase difference as part of the maximum likelihood or soft-output decoding of the underlying forward error correcting code. Therefore, by compensating the constant phase difference, hereafter we assume that $\alpha\mathsf{a}=\mathsf{b}\beta$. 

\begin{figure}
	\includegraphics[width=.5\textwidth,center]{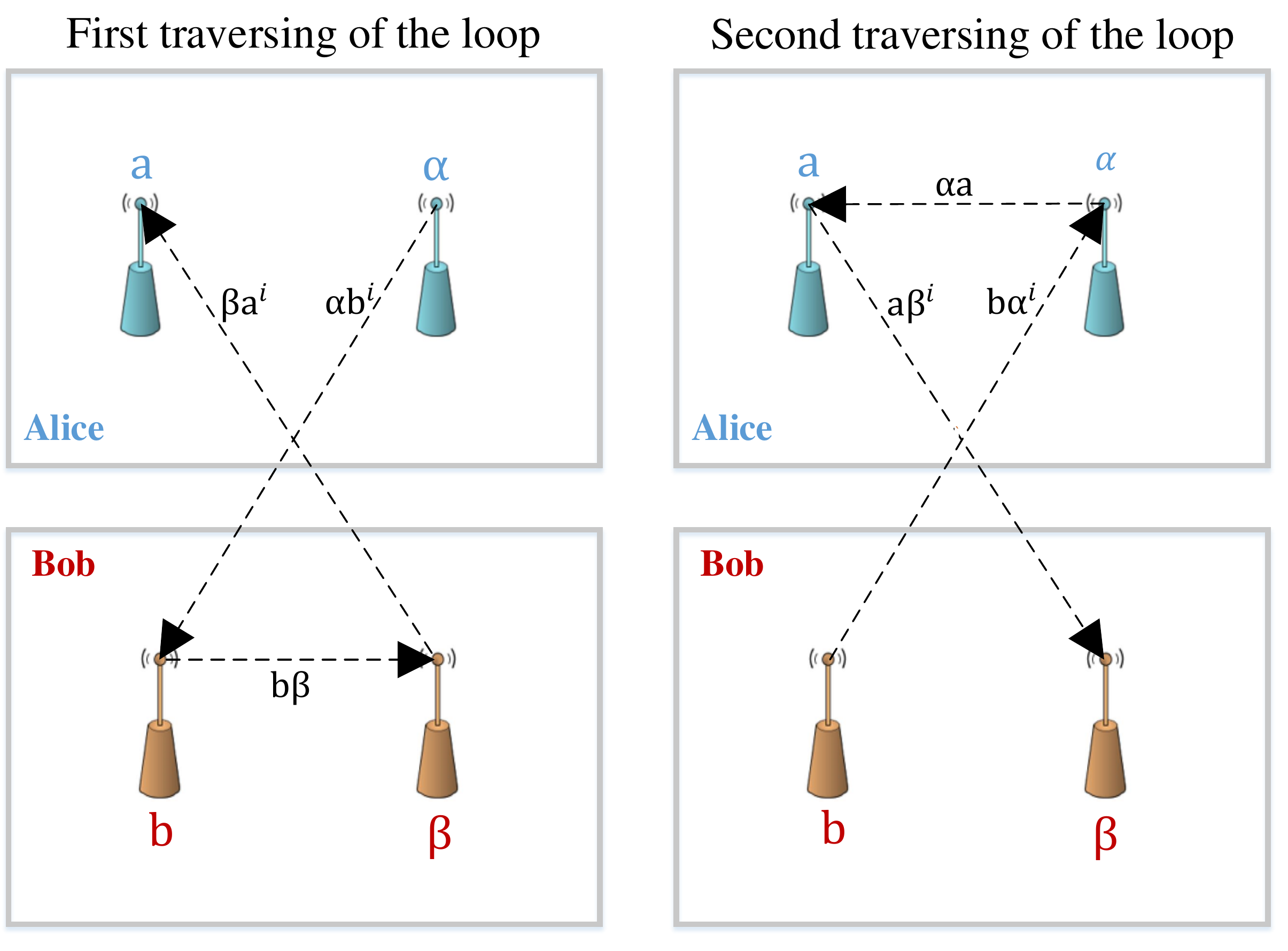}
	\caption{Traversing loops in four-antenna system. Note that $\alpha\mathsf{a}=\mathsf{b}\beta$ are independent of the mirror's state.}
	\label{fig:4-antenna-config}
\end{figure}
In the first traversing loop, Alice selects its mirrors at a random state $i$. Then, a pilot with random phase $\theta^i$, uniformly distributed in $[0, 2\pi)$, is sent from Alice to traverse the following loop: $\alpha \to \mathsf{b} \to \beta \to \mathsf{a}$ (note that Alice records $\theta^i$). This is referred to as the first traversing of the loop. Once the first loop is completed, Alice measures the change in the phase of the received signal with respect to that of original preamble, and records $\alpha\mathsf{b}^i \oplus \mathsf{b}\beta \oplus \beta\mathsf{a}^i$. Note that the initial random phase $\theta^i$ is canceled, and that $\mathsf{b}\beta$ is independent of the mirror state (as it is caused via an internal connection). Thereafter, the same loop initiated by Bob with random phase $\phi^i$, uniformly distributed in $[0, 2\pi)$, is traversed in a different direction, namely, $\mathsf{b} \to \alpha \to \mathsf{a} \to \beta$. This is referred to as the second traversing of the loop. Similarly, Bob measures the phase shift with respect to its original pilot to obtain $\mathsf{b}\alpha ^i \oplus \alpha\mathsf{a} \oplus \mathsf{a}\beta^i$. Also, as discussed earlier,  $\alpha\mathsf{a}=\mathsf{b}\beta$. The reciprocity of the channel implies that $\alpha\mathsf{b}^i=\mathsf{b}\alpha^i$ and $\beta\mathsf{a}^i=\mathsf{a}\beta^i$. Hence, the two legitimate nodes have measured the same phase shift which would be used as a shared phase value. To generate a new key, Alice changes the state of its mirrors to a new random state $j$ and again two traversing loops will be completed to establish a new shared phase value.   
\subsubsection{Analysis} 
Now, we show that an eavesdropper equipped with $n$ antennas is not capable of attaining the shared key between Alice and Bob.

	First, we find out what information Eve can obtain by its $k^{th}$ antenna when both Alice and Bob complete their traversing loops. It is assumed that Alice mirrors are in state $i$. 
	\\
	First traverse of the loop: 
	\begin{itemize}
		\item{$\alpha \to \mathsf{b}$: Eve measures  $\alpha\mathsf{e}_k^i \oplus \theta^i$, \hfill  (i)}
		\item{$\beta  \to \mathsf{a}$: Eve measures $\alpha\mathsf{b}^i \oplus \mathsf{b}\beta  \oplus \beta\mathsf{e}_k \oplus \theta^i$, \hfill  (ii)}
	\end{itemize}
	Second traverse of the loop: 
	\begin{itemize}
		\item{$\mathsf{b} \to \alpha$: Eve measures $\mathsf{be}_k \oplus \phi^i$,  \hfill  (iii)}
		\item{$\mathsf{a} \to \beta$: Eve measures $\mathsf{b}\alpha ^i \oplus \alpha\mathsf{a} \oplus \mathsf{ae}_k^i \oplus \phi^i$,  \hfill  (iv)}
	\end{itemize}
	for $1 \leq k \leq n$. Eve's observations obtained over its different antennas are independent of each other, and therefore they could be considered separately. Hence, we remove the Eve's antenna indexes in our analysis, i.e., we consider a single antenna at the Eve's side. Referring to the observations (i)-(iv), we express the summation of initial random phases with the phase of the channel between Eve and legitimate users as $\mathsf{m}_1^i = \alpha\mathsf{e}^i \oplus \theta^i$, $\mathsf{m}_2^i = \beta\mathsf{e} \oplus \theta^i$, $\mathsf{m}_3^i = \mathsf{be} \oplus \phi^i$, and $\mathsf{m}_3^i = \mathsf{ae}^i  \oplus \phi^i$.
	Based on Theorem \ref{uniform}, it is concluded that since both $\theta^i$ and $\phi^i$ are uniformly distributed in $[0, 2\pi)$, the phase values $\{\mathsf{m}_1^i,\mathsf{m}_2^i,\mathsf{m}_3^i,\mathsf{m}_4^i\}$ are also uniformly distributed in $[0, 2\pi)$. Furthermore, we assume that Eve is aware of the antenna structure utilized by both legitimate nodes, and thus she knows the value of $\alpha\mathsf{a}=\mathsf{b}\beta$. Therefore, Eve forms the following system of equations to find the value of $\alpha\mathsf{b}^i$.
	\begin{align} \label{sys_simple}
		& \alpha\mathsf{b}^i  \oplus \mathsf{m}_2^i=y_1^i \nonumber \\
		& \mathsf{b}\alpha ^i \oplus \mathsf{m}_4^i=y_2^i.
	\end{align}
	
	The system of equations in \eqref{sys_simple} has three unknowns---note that because of channel reciprocity $\alpha\mathsf{b}^i=\mathsf{b}\alpha^i$---and two equations; hence, it is under-determined \cite{datta2010numerical}. It is straight-forward that a system of equations of the form \eqref{sys_simple} (where the variables are in $[0, 2\pi)$ and the summations are modulo-$2\pi$) can be solved as follows: first, the system is solved as an ordinary system of linear equations, and second, the modulo-$2\pi$ of the unknown variables found in the first step will be the final solution for the original system. On the other hand, the coefficients of all the unknown variables in \eqref{sys_simple} are one, and in order for Eve to solve such a system, she needs to replace an unknown variable from one equation into the other equation. This will give Eve the summation of two unknown variables both having coefficient one. Thus, based on Theorem 1, Eve will face a phase ambiguity uniformly distributed in $[0, 2\pi)$ when determining any of unknown variables in \eqref{sys_simple}. This essentially means that Eve cannot extract any information about the shared key based on her observations.   
\section{Modulation} \label{modulation}
In this section, a simple modulation technique for a secure data exchange between legitimate users utilizing the proposed key sharing method is presented. Assume that $\Phi$ is a secret common phase value between the legitimate nodes that is obtained by one of the two systems proposed in this paper. Suppose we have $2^m$-PSK symbols as our plain text (which could represent the message, or a bit stream to be used as key). At the receiver, we rotate each $2^m$-PSK constellation point $X$ with the common phase value $\Phi$ which is equivalent to modulo-$2\pi$ addition of $X$ and $\Phi$:
$Y=X\oplus\Phi$. 	
Any eavesdropper who observes $Y$, cannot extract any information about $X$ because $\Phi$ is a random variable with uniform distribution in $[0,2\pi)$, and based on Theorem 1, $Y$ is also a random variable uniformly distributed  in $[0,2\pi)$ and independent from $X$, leading to $I(X;Y) = H(Y)-H(Y|X)=H(Y)-H(\Phi)=0$. At the receiver side, the legitimate receiver knowing the common phase value de-rotates each constellation point accordingly and recovers the message which is equivalent to modulo-$2\pi$ subtraction of $\Phi$ from $Y$ (not considering noise for simplicity of notation): $X=Y \ominus \Phi$.
This process allows legitimate parties to share $m$ secure bits per common phase value. In the following subsection, an example of key exchange when the legitimate users utilize Quadrature Phase Shift Keying (QPSK) is elaborated, and furthermore, the effect of FEC is discussed.

\subsection{Example for Key Exchange}
Assuming there are $\Theta$ shared phases, and the modulation for the transmission of key is QPSK, with a rate 1/2 FEC. One of the two nodes, acting as the master, generates $\Theta$ random bits (representing a messege to be securely communicated, or to be used as the key), passes these bits through FEC encoder to generate $2\Theta$ bits, maps the resulting $2\Theta$ bits to $\Theta$ QPSK constellations, rotates each constellation with one of the shared phase values, and transmits the rotated constellations to the other node. The receiving node will first de-rotate subsequent constellations with  its local copies of the corresponding phase values, and then demodulate the constellation and decode the FEC to extract the $\Theta$ bits of information. Note that, in this example, decoding the FEC is equivalent to decoding it in an AWGN channel. In addition, it is seen that the redundant bits added by using FEC are also transmitted by QPSK constellations. Since these constellations are rotated by one of the shared random phase values, they form an independent constellation stream at the receiver, each having a phase ambiguity uniformly distributed in $[0,2\pi)$. As discussed in Section \ref{secrecy}, this in turn concludes that the redundant bits---those added by FEC---do not disclose any information to Eve.  

\section {Empirical Results} \label{results}
\begin{figure} 
	\centering
	\subfloat[\label{warp}]{%
		\includegraphics[width=0.35\linewidth]{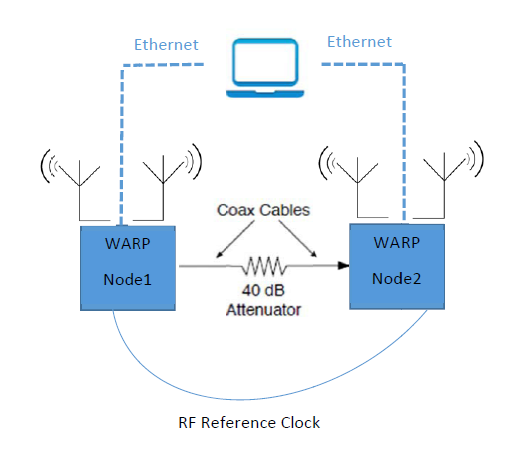}}
	\subfloat[\label{twophase2}]{%
		\includegraphics[width=0.4\linewidth]{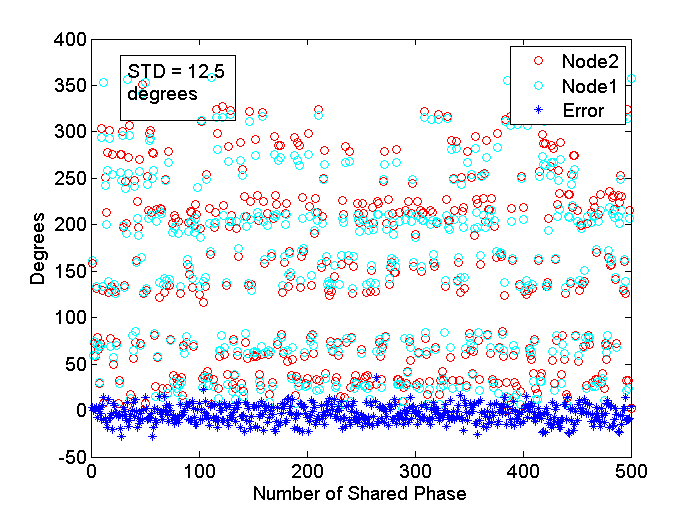}}
	\caption{(a) Setup for implementation on WARP platform. (b) Over-the-air system performance.}
\end{figure}
The structure is implemented on two WARP1 software radio platforms operating at a carrier of 2.4Ghz. Fig.\,\ref{warp} depicts the implementation setup block diagram. Each RF mirror is a set of 6 conductor patches connected to each other with high frequency PIN diodes acting as RF switches. If these PIN diodes are ON, the conductor patches are connected and act as a perfect conductor wall reflecting the incident wave, whereas if they are OFF, the disjoint conductor patches act as a set of parasitic elements, resonating with the antenna, and in turn radiate externally.  As mentioned earlier, due to space limitations, the two nodes were placed within the same room at a distance of a few meters. As a result, there is a LOS component which contributes to all channel states including the reference state, and makes the constellation points  spreading around a center point. LOS is a constant and predictable component, adding a bias to each constellation preventing the random phase values to be uniformly distributed over $[0,2\pi)$. To overcome this shortcoming, the bias due to the LOS component is subtracted from the measurements. In our experiments, transmit power is about 10dBm, which is partially wasted (absorbed) within the mirror enclosure prior to propagating externally. The variance of phase error could be reduced by using a higher transmit power, and/or by averaging over multiple transmissions. Phase error can be also reduced by discarding points with low signal magnitude, i.e., when the magnitude of the multi-path fading is too small, and/or a higher amount of energy is absorbed within the enclosure. Fig.\,\ref{twophase2} shows an example for over-the-air phase measurement in which LOS is removed and about $20\%$ of points of lowest energy are discarded.
 
\section{conclusion} \label{conclusion}
This article presents a practical method that exploits phase reciprocity in wireless channels for perfectly secure key exchange. The proof of perfect secrecy granted by the proposed method is provided. Also, two different antenna configurations with their respective protocols of key sharing are proposed.  Unlike known techniques for PLS, which have unsolved challenges in terms of removing possible errors between the shared keys (without disclosing information), the proposed structure relies on well established techniques in wireless transmission such as ``synchronization", ``channel phase measurement" and  ``Forward Error Correction''. The hardware components are simple, low cost and the rate of key generation is up to the limits governed by Nyquist sampling theorem. 

\bibliographystyle{IEEEtran}
\bibliography{refs}

\end{document}